\documentclass[a4paper,twocolumn,10pt,accepted=2024-02-28]{quantumarticle}
\pdfoutput=1
\usepackage[utf8]{inputenc}
\usepackage[english]{babel}
\usepackage[T1]{fontenc}
\usepackage[colorlinks,citecolor=blue,linkcolor=magenta,urlcolor=blue]{hyperref} 
\usepackage{amsmath, amssymb}
\usepackage{graphicx}
\usepackage{graphics}
\usepackage{mathtools}
\usepackage{amsthm}
\usepackage{mathrsfs}
\usepackage{color}
\usepackage{enumitem}

\newcommand{\bea}{\begin{eqnarray}}
\newcommand{\eea}{\end{eqnarray}}

\def\C{\hbox{$\mit I$\kern-.7em$\mit C$}}
\def\R{\hbox{$\mit I$\kern-.6em$\mit R$}}
\def\N{\hbox{$\mit I$\kern-.6em$\mit N$}}

\def\ket#1{|#1\rangle}
\newcommand{\one}{\mbox{$1 \hspace{-1.0mm}  {\bf l}$}}
\def\tr{\mathrm{tr}}
\def\bra#1{\left< #1\right|}

\newcommand{\proj}[1]{\ket{#1}\bra{#1}}
\newcommand{\kb}[2]{\ket{#1}\bra{#2}}

\newtheorem{theorem}{Theorem}

\newtheorem{lemma}[theorem]{Lemma}

\newtheorem{observation}[theorem]{Observation}

\begin{document}

\title{Transformations in quantum networks via local operations assisted by finitely many rounds of classical communication}

\author{Cornelia Spee}
\affiliation{Institute for Theoretical Physics, University of Innsbruck, Technikerstra{\ss}e 21A, 6020 Innsbruck, Austria}
\affiliation{Institute for Quantum Optics and Quantum Information (IQOQI),
Austrian Academy of Sciences, Boltzmanngasse 3, 1090 Vienna, Austria}
\author{Tristan Kraft}
\affiliation{Institute for Theoretical Physics, University of Innsbruck, Technikerstra{\ss}e 21A, 6020 Innsbruck, Austria}
\affiliation{Naturwissenschaftlich-Technische Fakult\"at, 
Universit\"at Siegen, Walter-Flex-Stra{\ss}e 3, 57068 Siegen, Germany}         

\begin{abstract}
Recent advances have led towards first prototypes of quantum networks in which entanglement is distributed by sources producing bipartite entangled states. This raises the question of which states can be generated in quantum networks based on bipartite sources using local operations and classical communication. In this work, we study state transformations under finite rounds of local operations and classical communication (LOCC) in networks based on maximally entangled two-qubit states. We first derive the symmetries for arbitrary network structures, as these determine which transformations are possible. Then, we show that contrary to tree graphs, for which it has already been shown that any state within the same entanglement class can be reached, there exist states which can be reached probabilistically but not deterministically if the network contains a cycle. Furthermore, we provide a systematic way to determine states which are not reachable in networks consisting of a cycle. Moreover, we provide a complete characterization of the states which can be reached in a cycle network with a protocol where each party measures only once, and each step of the protocol results in a deterministic transformation. Finally, we present an example which cannot be reached with such a simple protocol, and constitutes, up to our knowledge, the first example of a LOCC transformation among fully entangled states requiring three rounds of classical communication.
\end{abstract}

\maketitle

\section{Introduction}
In the last decades, a lot of effort has been devoted to realizing the first prototypes of a quantum network~\cite{Kimble2008,Wehner2018}, based on various physical platforms~\cite{Cirac1997,Duan2010,Reiserer2015}. Quantum networks offer the possibility for quantum information processing tasks such as long-distance quantum communication~\cite{Duan2001}, or distributed quantum computing~\cite{Cirac1999,Spiller2006}. Moreover, they are also interesting from a theoretical point of view since they require novel tools to study the correlations that can arise from such networks, which is in itself a challenging problem, and to evaluate their performance~\cite{Elkouss2020} (and Refs. therein). Nevertheless, a lot of progress has recently been made concerning the characterization of correlations~\cite{Gisin2020} (and Refs. therein) and entanglement~\cite{Kraft2020triangle,Navascues2020,Luo2020,Aberg2020,Kraft2020,Hansenne2022}.

In quantum networks, entanglement is distributed by multiple sources, rather than a single source. Subsequently, the entanglement created in the network can be manipulated, e.g., by means of entanglement swapping~\cite{teleport}, and can thus be spread over the entire network~\cite{Acin2007}. Entanglement swapping is an instance of a much larger and important class of local operations assisted by classical communication (LOCC). The general question consists of identifying which state transformations are possible via LOCC and which are the most useful states under such a restriction. In the bipartite case, this question is answered by the Nielsen majorization criterion~\cite{Nielsen}. In case the transformation is not required to be deterministically, one arrives at the so-called stochastic LOCC (SLOCC) operations. For instance, in the case of three qubits, it is known that there exist only six classes of states that are equivalent under SLOCC~\cite{Duer}. For four parties, the number of SLOCC classes is already infinite~\cite{Verstraete}.

Despite the fact that the class of LOCC operations is of fundamental importance for many quantum information tasks it is notoriously difficult to characterize (see, e.g.,~\cite{Donald2002,Chitambar2011,Chitambar2012,WinterLOCC,Cohen2017}).
Despite many difficulties, significant progress has been made, e.g., in the case where only finitely many rounds of classical communication (finite-round LOCC) are considered~\cite{Turgut1,Turgut2,finiteroundLU,finiteroundLU2,Vicente2013,Schwaiger2015,Spee2016,Hebenstreit2016}.

\begin{figure}[t]
\centering
\includegraphics[width=\columnwidth]{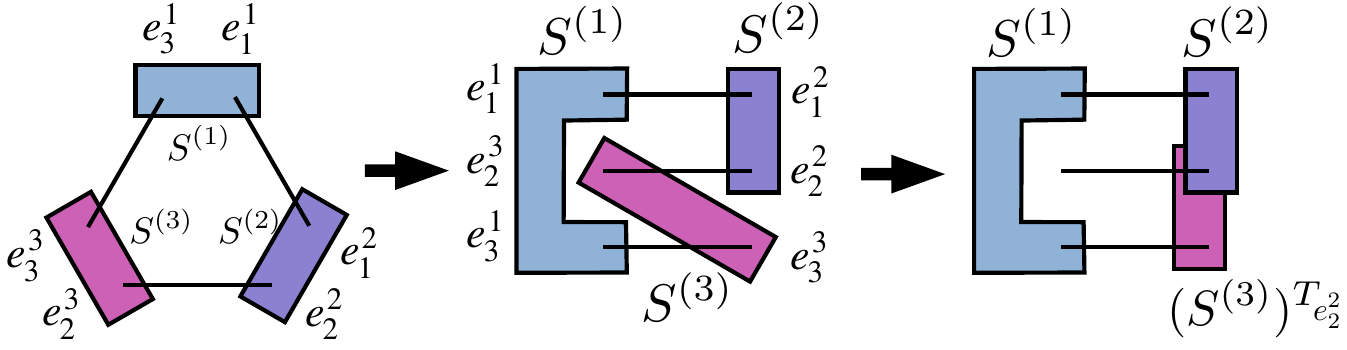} 
\caption{Schematic illustration of the procedure that allows to identify the symmetries of the triangle network. the first and the second picture are identical up to rearrangement of the parties and nodes. In order to transform the second picture into the third one we used that $Y\otimes \one\ket{\Phi^+}=\one\otimes Y^T\ket{\Phi^+}$. This, together with the fact that all symmetries of the $\ket{\Phi^+}$ state are of the form $Y\otimes Y^{-T}$, proves that Eq.~\eqref{eq_symmetry} holds.
\label{fig:symm}}
\end{figure}

In this work, we study the transformation of entanglement in quantum networks under LOCC. We emphasize that in case the parties have access to many copies of maximally entangled states they can use teleportation in order to realize an arbitrary entangled state. This might, however, involve a potentially large number of rounds of communication. The fact that quantum memories cannot provide long storage times motivates us to consider the more realistic scenario where sources distribute single copies of bipartite maximally entangled states according to some predefined network structure and then perform an LOCC transformation involving only a finite number of rounds.

In Ref.~\cite{Hayata}, it has already been shown that if the graph is a tree graph, then any state within this SLOCC class can be realized. As we will show, for any network that contains a cycle, i.e., that is not a tree graph, this no longer holds true. We will determine pure states which can/cannot be reached within such a scenario. First, we will show that surprisingly, in such a network, all local symmetries are of a very specific form. This allows us to provide a systematic approach to rule out a large class of states which can never be realized within networks consisting of a cycle. Moreover, we identify all states which can be obtained in such networks under the restriction that each party performs at most one non-trivial measurement. Finally, we show that there exist states which cannot be reached with a protocol where each party measures only once, and which require three rounds of communication.

\section{Preliminaries}
Any quantum network can be represented by a graph $G$ consisting of nodes $1,\dots,n$ which are connected by edges $e\in E$. Each edge $e$ corresponds to a two-qubit state $\ket{\psi}_e$ that is shared between the two nodes it connects, i.e., $e\equiv \{e_j^{i},e_j^{k}\}\in E$, where $j$ labels the source, and $i,k$ label the nodes that are connected by the source $j$. We will use the convention that lower indices refer to parties. Whenever additional indices are required, such as summation or source indices, the party index will be superscripted. Moreover, we will, whenever it is convenient, denote the identity matrix and the Pauli matrices $\sigma_x,\sigma_y,\sigma_z$ by $\sigma_i$, with $i=0,1,2,3$. We will assume that, initially, each edge corresponds to a maximally entangled state $\ket{\Phi^+}_e=\frac{1}{\sqrt{2}}(\ket{00}+\ket{11})$. Thus, the initial state of the network can be written as
\begin{equation}\label{eq:initial}
\ket{\psi_{\rm{ini}}}=\bigotimes_{e\in E} \ket{\Phi^+}_e.
\end{equation}
The number of qubits that each node holds is determined by the number of connected edges, that is, its degree in the graph $G$.

In the following, we will restrict to transformations among pure states [up to local unitary (LU) transformations] that preserve the ranks of the local reduced density matrices.
In this case, it is not possible to reach any state outside of the SLOCC class~\cite{Duer}
\begin{equation}\label{eq:final}
h_1\otimes \ldots \otimes h_n \prod_{e\in E} \ket{\Phi^+}_e
\end{equation}
via a finite-round LOCC protocol, where $h_i$ is an invertible operator acting on all the particles node $i$ has received from the sources. As LOCC and LOCC$_\mathbb{N}$ transformations have a rather complicated structure, it is sometimes useful to consider larger classes of transformations for which the possible transformations are easier to characterize. For instance, LOCC operations are strictly contained in the set of \emph{separable} operations (SEP). A completely positive trance preserving (cptp) map $\Lambda$ is called separable if all its Kraus operators are separable, i.e., $\Lambda(\varrho )=\sum_i K_i\varrho K_i^\dagger$, with $K_i=K_i^1\otimes\dots\otimes K_i^n$. It is well-known that $LOCC_\mathbb{N}\subsetneq SEP$~\cite{WinterLOCC} (and Refs. therein), and thus, if a transformation is not possible via SEP, it is also not possible via LOCC.

For LOCC$_\mathbb{N}$ transformations among fully entangled pure states, which we also consider here, it is known that they are contained in the set of separable operations for which all Kraus operators are invertible, so-called SEP$_1$ transformations~\cite{measuroutcdonotexist}. We will use the following SEP$_1$ condition~\cite{Gour,Gour2} for the described scenario, which is a necessary condition on finite-round LOCC transformations to exist:
\begin{equation}\label{SEP1}
\sum_i p_i S_i^\dagger H S_i= r \one,
\end{equation}
where $\{p_i\}$ is a probability distribution, and $r\in\mathbb{R}$. Moreover, in this context and in the following,
\begin{equation}
H=\otimes _j H_j=\otimes_j h_j^\dagger h_j
\end{equation}
is determined by the final state in Eq.~\eqref{eq:final}, and $S_i$ are  local symmetries of the initial state in Eq.~\eqref{eq:initial}. The local symmetries are all $S=\otimes S^{(j)}$, with $S^{(j)}\in GL(d_j)$ and $d_j$ being the local dimension of party $j$, for which
\begin{equation}
S\ket{\psi_{\rm ini}}=\ket{\psi_{\rm ini}}.
\end{equation}
The set of all local symmetries we will denote by $\mathbf{S} (\ket{\psi_{\rm ini}})$. Intuitively, the condition in Eq.~\eqref{SEP1} guarantees that there exist Kraus operators that define the correct cptp map in SEP$_1$ that can achieve the transformation from the initial state $\ket{\psi_{\rm ini}}$ to the final state $h\ket{\psi_{\rm ini}}$~\cite{Gour,Gour2}. Hence, it is the symmetries of the initial state which allow the transformation to be deterministic.

In~\cite{MPS}, the symmetries and SLOCC transformations of translationally invariant matrix product states (which include certain network structures) have been characterized. In the following, we characterize the symmetries of arbitrary networks of bipartite sources that distribute maximally entangled two-qubit states.

\section{Symmetries of Quantum Networks}
For simplicity, we will only provide the explicit proof for the symmetries of the triangle network, as shown in Fig.~\ref{fig:symm}. However, the proof can be straightforwardly generalized to arbitrary network structures.

In order to derive the symmetries, we will use the property of $\ket{\Phi^+}$ that $Y\otimes \one\ket{\Phi^+}=\one\otimes Y^T\ket{\Phi^+}$. With this we have that
\begin{align}\nonumber
&S^{(1)}\otimes S^{(2)}\otimes S^{(3)} \ket{\Phi^+}_{e^1_1 e_1^2}\ket{\Phi^+}_{e_2^2e_2^3}\ket{\Phi^+}_{e_3^3e_3^1}=\\
&S^{(1)}\otimes \one_{e_2^3}\otimes\bigg[(S^{(2)}\otimes \one_{e_3^3}) \{(S^{{(3)}})^{T_{e_2^2}} \otimes \one_{e^2_1}\}\bigg]\ket{\Phi^+}^{\otimes 3},
\end{align}
where $T_{e_2^2}$ denotes the partial transpose on $e_2^2$ (see also Fig.~\ref{fig:symm} for a graphical illustration). Using that the well known symmetries of $\ket{\Phi^+}$ are $Y\otimes Y^{-T}$ this implies that
\begin{align}\label{eq_symmetry}
(S^{(2)}\otimes \one_{e^3_3}) [(S^{{(3)}})^{T_{e_2^2}} \otimes \one_{e^2_1}]=(S^{(1)}_{e_1^2e_3^3})^{-T}\otimes\one_{e^2_2}.
\end{align}
Remarkably, this condition can only be satisfied if $S^{(2)}$ and $(S^{{(3)}})^{T_{e_2^2}}$, and thus, also $S^{(3)}$, factorize. More precisely, one can show that 
\begin{equation}\label{eq_fact_mt}
X_{AB}Z_{BC}=Y_{AC}\otimes \one_B,
\end{equation}
with system $B$ being two-dimensional and $X, Z$ invertible operators, implies that $X$ and $Z$ factorize with respect to the splitting $A|B$ and $B|C$ respectively (see Lemma~\ref{Lemma_factorize} in Appendix~\ref{Appendix_A}). Note that, therefore, $S_2$ and $S_3$ factorize for any symmetry of the triangle. The triangle is invariant under cyclic permutation of the parties, and therefore, via an analogous argumentation, one can show that also $S_1$ has to factorize for any symmetry. Moreover, we emphasize that the same argument can straightforwardly be extended to any network structure.

To see this, observe that by construction, any qubit is maximally entangled to exactly one other qubit, which is held by a different party. Moreover, all qubits that belong to one party are all entangled to different parties. Thus, if one considers any party-local symmetry on one party, one can repeatedly use that $Y\otimes \one\ket{\Phi^+}=\one\otimes Y^T\ket{\Phi^+}$ to make this symmetry overlap with another symmetry on exactly one qubit. Then one employs Eq.~\eqref{eq_fact_mt} to see that the operators factorize (see Fig.~\ref{LOCCNetworkTrick} for a simple example). Repeating this procedure for all qubits on any party, one sees that all symmetries have to be local symmetries. Finally, if the graph $G$ that describes the connectivity of the network contains leaves, i.e., nodes that only have a single neighbour, one arrives in a situation where Eq.~\eqref{eq_fact_mt} becomes trivial, i.e., $X_{B}Z_{BC}=Y_{C}\otimes \one_B$. In this case, one can directly observe that $Z_{BC}$ factorizes by multiplying with $X_{B}^{-1}$ from the left.

With this, we have that there can only exist symmetries acting locally on $e^i_j$ and $e^k_j$ for any $e_j\equiv \{e^i_j,e^k_j\}\in E$.
Using that the local symmetries of $\ket{\Phi^+}$ are $X\otimes X^{-T}$ we obtain that the set of symmetries is given by $\bigotimes_{e\in E} (X_e\otimes X_e^{-T})_e$. Note that here $X_e$ can be different for any edge $e$. This results in the following observation.
\begin{observation}
For any network structure the local symmetries of $ \prod_{e\in E} \ket{\Phi^+}_e$ are given by $\bigotimes_{e\in E} (X_e\otimes X_e^{-T})_e$.
\end{observation}
Hence, only symmetries acting locally on single qubits can contribute and enable transformations.

\begin{figure}[t]
\includegraphics[width=\columnwidth]{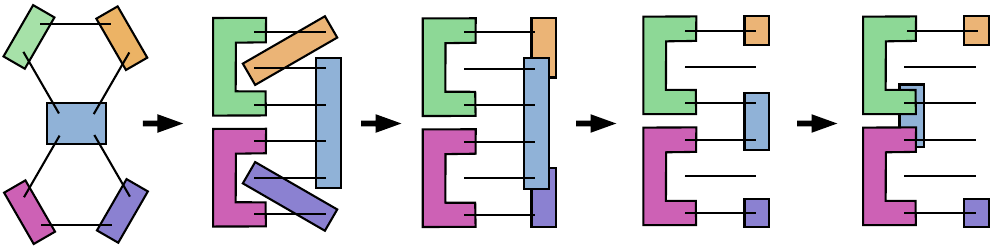} 
\caption{Schematic illustration of the procedure that allows to identify the symmetries for the example of the double-triangle network. In order to transform one picture into the other $Y\otimes \one\ket{\Phi^+}=\one\otimes Y^T\ket{\Phi^+}$ is used, as well as that Eq.~(\ref{eq_fact_mt}) implies that the operators factorize. From the last picture it follows that all symmetries factorize.
\label{LOCCNetworkTrick}}
\end{figure}

\section{States which cannot be prepared via finite round LOCC transformations in networks}
Within a tree graph network, any state in the corresponding SLOCC class is reachable~\cite{Hayata}. Here, we will show that if a network contains a cycle, this is no longer true. For convenience, we will consider here and in the following as initial state $ \bigotimes_{e\in E} \ket{\Psi^-}_e$, which is local unitary (LU) equivalent to $ \bigotimes_{e\in E} \ket{\Phi^+}_e$, and therefore has the symmetries $\bigotimes_{e\in E} (X_e\otimes X_e)_e$ with $\det (X_e)=1$. We make the following observation.
\begin{observation}
An example for a class of states which  cannot be reached via finite round LOCC from the initial network state containing a cycle is given by $(\otimes_{i\in cycle} h_i\otimes_{j\not\in cycle} \one) \bigotimes_{e\in E} \ket{\Psi^-}_e$ with $H_i=(\one/4+\alpha_i \sigma_z\otimes\sigma_z)\otimes_{\not\in cycle} \one $ acting non-trivially only on the two qubits which are part of the cycle and $\prod_{i\in cycle}\alpha_i> 0$.
\end{observation}
\begin{proof} 
Consider the projection of Eq.~\eqref{SEP1}, with $H_i$ being of this form, on $\ket{\Psi^-}_e$ for all edges except one edge that is part of the cycle, i.e., on $\bigotimes_{E\backslash \tilde{e}\in cycle} \ket{\Psi^-}_e$. One obtains that in order for the transformation to be possible it has to hold that
 \begin{equation}\label{SEP1_proj}
\sum_i p_i (X_i \otimes X_i)^\dagger J (X_i \otimes X_i)= r \one,
\end{equation}
with $J=\one/4^N+\prod_{j\in cycle}\alpha_j\sigma_z\otimes\sigma_z$ and $N$ the number of parties within the cycle. From the projection of this equation on $\ket{\Psi^-}$ one obtains that $r=1/4^N-\prod_{j\in cycle}\alpha_j$.

Considering then the trace leads to an equation of the following form
 \begin{multline}
\sum_i p_i\frac{1}{4^N}\Big\{[|a_i|^2+|c_i|^2+|b_i|^2+|d_i|^2]^2-4\Big\}\\
+\Big\{[|a_i|^2+|c_i|^2-(|b_i|^2+|d_i|^2)]^2\\
+4\Big\}\prod_{j\in cycle}\alpha_j=0,
\end{multline}
where $a_id_i-b_ic_i=1$ for any $i$. As $[|a_i|^2+|c_i|^2+|b_i|^2+|d_i|^2]^2\geq 4$ for $a_id_i-b_ic_i=1$ and $[|a_i|^2+|c_i|^2-(|b_i|^2+|d_i|^2)]^2+4>0$ this equation can not be fulfilled if $\prod_{j\in cycle}\alpha_j>0$.
\end{proof}
We note that, as the proof relies on Eq.~\eqref{SEP1}, the class of states defined above is also not reachable by SEP$_1$ transformations. Before characterizing all states which can be reached with simple LOCC protocols in cycle networks, we present a systematic way to decide whether or not a state is reachable for many cases.

\section{Systematic way to determine finite round LOCC transformations which are not possible in a cycle network}\label{sec:systematic}
In the following, we will restrict to networks which correspond to a cycle of size $N$ (which we will call cycle networks).
Moreover, we will use the following necessary and sufficient condition for a transformation to be possible if party $k$ performs a generalized measurement, communicates the outcome to the other parties, and these apply LUs. Note that this corresponds to a single round within a LOCC protocol. Such a transformation is solely possible among pure states if the following conditions are satisfied. Firstly, the initial state is of the form
\begin{equation}
h_1\otimes \ldots \otimes g_k\otimes h_{k+1}\ldots h_N\ket{\Psi}
\end{equation}
and the final state is (up to LUs)
\begin{equation}
h_1\otimes \ldots \otimes h_k\otimes h_{k+1}\ldots h_N\ket{\Psi}.
\end{equation}
Secondly, there must exist local symmetries $S_i=\bigotimes_j S_i^{(j)}\in \mathbf{S}(\ket{\Psi})$ and probabilities $\{p_i\}$ such that
\begin{equation}\label{commute}
(S_i^{(j)})^\dagger H_jS_i^{(j)}\propto H_j
\end{equation}
for parties $j=1, \ldots,k-1, k+1,\ldots N$, and for party $k$
\begin{equation}\label{finiteLOCCk}
\sum_i p_i (S_i^{(k)})^\dagger H_k S_i^{(k)}=r G_k= r g_k^\dagger g_k,
\end{equation}
for some $r\in\mathbb{R}$~\cite{finiteroundLU,finiteroundLU2,finiteroundgeneral}. The transformation is non-trivial, i.e., the initial and the final state are not LU equivalent, if
\begin{equation}\label{eq:non-trivial}
(S_i^{(k)})^\dagger G_k S_i^{(k)}\not\propto H_k
\end{equation}
for all symmetries for which Eq.~\eqref{commute} holds true.

Note that for any finite round LOCC protocol transforming a pure state into some other pure state, the final round is exactly of the form described above, i.e., it has to transform deterministically one pure state into another within a single round. Hence, for any state which is reachable via a finite-round LOCC protocol, there have to exist symmetries for which Eq. (\ref{commute}) holds true for $N-1$ parties and for at least one symmetry $S_i$ for which Eq. (\ref{commute}) holds
\bea\label{notcommute}(S_i^{(k)})^\dagger H_kS_i^{(k)}\not\propto H_k\eea for the remaining party (in order for an operator $G_k\not\propto H_k$ to possibly exist)~\cite{finiteroundgeneral}. 
In order to decide whether a given $H$ fulfills this requirement one can first determine the local operators $X_j\otimes Y_j$ which fulfill Eq. (\ref{commute}) locally for a given $H_j$, i.e.,
\begin{equation}\label{commute02}
(X_j\otimes Y_j)^\dagger H_j(X_j\otimes Y_j)\propto H_j.
\end{equation}
Then one considers for these local operators $\bigotimes_{j\neq k} X_j\otimes Y_j$ and determines whether there exists a symmetry such that $\bigotimes_{j\neq k}S_i^{(j)}=\bigotimes_{j\neq k} (X_j\otimes Y_j)$. In case this is possible, one can consider all the symmetries $S_i$ for which this is true and check whether Eq. (\ref{notcommute}) holds.

In the following, we present a systematic way to evaluate for a given $H_j$ which operators fulfill Eq. (\ref{commute02}) locally. In order to do so, we use ideas which have also been used in Ref.~\cite{Spee2016} to determine the symmetries of four-qubit states. Eq.~\eqref{commute02} holds true if and only if 
\begin{equation}\label{commute2}
(\tilde{X}_j\otimes \tilde{Y}_j)^\dagger \tilde{H}_j(\tilde{X}_j\otimes \tilde{Y}_j)\propto \tilde{H}_j
\end{equation}
with $\tilde{H}_j=(M\otimes \bar{M})^{\dagger} H_j(M\otimes \bar{M})$ and $\tilde{X}_j\otimes \tilde{Y}_j=M^{-1}X_jM\otimes \bar{M}^{-1}Y_j\bar{M}$ for any $M, \bar{M}$ invertible. We then choose  $M\otimes \bar{M}$ such that $\tilde{H}_j$ is Bell diagonal denoted by $\tilde{H}_j^B$. Note that it has been shown~\cite{Belldiagonal} that this is always possible. We will call in the following $\tilde{H}_j^B$ the local standard form of $H_j$ and we obtain
\begin{equation}\label{commute3}
(\tilde{X}_j^B\otimes \tilde{Y}_j^B)^\dagger \tilde{H}_j^B(\tilde{X}_j^B\otimes \tilde{Y}_j^B)\propto \tilde{H}_j^B.
\end{equation}
Transforming then $\tilde{H}_j^B$ and $\tilde{X}_j^B\otimes \tilde{Y}_j^B$ into the magic basis, i.e., $ H_j^{mb}=U^\dagger\tilde{H}_j^BU$ and $O_j=U^\dagger \tilde{X}_j^B\otimes \tilde{Y}_j^B U$ where $U=\kb{\Phi^+}{00}-i \kb{\Phi^-}{01}+\kb{\Psi^-}{10}-i\kb{\Psi^+}{11}$, we obtain that $H_j^{mb}$ is a diagonal matrix with real non-zero entries (as $H_j$ is a positive definite matrix) and $O_j$ is a special orthogonal matrix. With this we have that 
\begin{equation}\label{commute4}
A_j=O_j^\dagger H_j^{mb} O_j\propto H_j^{mb}.
\end{equation}
This implies that 
\begin{equation}\label{commute5}
A_j^TA_j=(O_j)^T (H_j^{mb})^TH_j^{mb} O_j= (H_j^{mb})^TH_j^{mb}.
\end{equation}
Note that as this is a similarity transformation the set of eigenvalues has to stay invariant. As the eigenvalues of $ (H_j^{mb})^TH_j^{mb}$ have to be real and positive this implies that there is no proportionality factor. Alternatively, Eq. (\ref{commute5}) can be written as $[(H_j^{mb})^TH_j^{mb}, O_j]=0$ and it is straightforward to solve this equation. From the resulting matrices $O_j$ one obtains potential candidates for $\tilde{X}_j^B\otimes \tilde{Y}_j^B=UO_jU^\dagger$. We then insert them into Eq. (\ref{commute3}) and evaluate the ones that indeed leave $H_j^B$ invariant up to a factor. Note that if one knows the local operators $\tilde{X}_j^B\otimes \tilde{Y}_j^B$ which fulfill Eq. (\ref{commute3}) for the local standard forms one can directly calculate the ones which fulfill Eq. (\ref{commute02}) for a given $H_j$ via the relations
\begin{equation}
H_j= (M^{-1}\otimes \bar{M}^{-1})^\dagger \tilde{H}_j^B(M^{-1}\otimes \bar{M}^{-1})
\end{equation}
and
\begin{equation}
X\otimes Y=M\tilde{X}_j^BM^{-1}\otimes \bar{M}\tilde{Y}_j^B\bar{M}^{-1}.
\end{equation}
The explicit symmetries for all local standard forms can be found in Appendix~B. Using the method presented in \cite{Belldiagonal} one can identify the matrices $M$ and $\bar{M}$ which relate $H_j$ to its local  standard form and with this it is easy to determine the matrices $X$ and $Y$ which fulfill Eq.~(\ref{commute02}). Determining from the so obtained operators whether there exists a way of concatenating them to form a symmetry on $N-1$ party allows to straightforwardly employ the necessary conditions in Eqs.~(\ref{commute}) and (\ref{notcommute}) to identify states that cannot be reached via finite round LOCC protocols in cycle networks.
In the following, we will use Eqs.~(\ref{commute}) and (\ref{finiteLOCCk}) also to determine the class of states which can be reached with a LOCC protocol where each party measures solely once.

\section{A simple class of finite round LOCC protocols in cycle networks}
In case the necessary conditions in the previous section are fulfilled, a transformation may be possible.  In particular, these conditions provide a complete characterization of states which can be realized with protocols where each party measures once in a cycle network. In order for a state to be reachable from the network of maximally entangled states with a protocol where each party measures only once, the conditions in Eqs. (\ref{commute}) and (\ref{finiteLOCCk}) have to be fulfilled for some choice of $k$ and $G_k=A^\dagger A\otimes B^\dagger B$ for some $A, B\in SL(2)$. If they do not hold true, the state cannot be reached via such a protocol, as can be seen as follows. For a successful transformation, the initial state in the final round has to be of the form $\bigotimes_{i\neq k} h_i\otimes g_k\ket{\Psi}$ where $k$ is the party measuring in the final round. It is easy to see that for the considered class of protocols, $g_k$ has to be of the form $A\otimes B $, as party $k$ did not perform any measurement yet and the only way how one may obtain a non-trivial $g_k$ is via the symmetries. However, if the conditions are not fulfilled the transformation in the final round cannot be deterministic. In case they hold true, one can devise a protocol implementing the transformation. This results in the following observation.
\begin{observation}\label{obs:observation3}
A state $h\ket{\Psi}$ can be reached from the initial network via LOCC, where each party measures only once, and each step of the protocol is deterministic, if and only if the conditions in Eqs.~(\ref{commute}),(\ref{finiteLOCCk}) and (\ref{eq:non-trivial}) are fulfilled for  some choice of $k$ and $G_k=A^\dagger A\otimes B^\dagger B$ for some $A, B\in SL(2)$.
\end{observation}
If the conditions in Observation 3 hold, the following protocol allows us to reach the state $h\ket{\Psi}$ deterministically. First, note that we consider in the following the SLOCC operator $H$ for which
\begin{equation}
\tr_{e_{j-1}^{j}}(H_j)= 2 \one
\end{equation}
holds for $j\in\{3,\ldots,N-1\}$ and
\begin{equation}
\tr_{e_{1}^{2}}(B^{-\dagger}_{e_{1}^{2}}H_2B^{-1}_{e_{1}^{2}})= 2 \one.
\end{equation}
Here we use that $S^\dagger H S$ with $S\in \mathbf{S}(\Psi)$ defines the same state as $H$ (see also \cite{Vicente2013}). Hence, there is some ambiguity in the choice of H and we will make use of this to ensure the desired property. Note that for any state in a cycle network, it is possible to transform a given $H$ in such a form (if  we choose to normalize the operator accordingly).  We further denote here without loss of generality $k=1$ and enumerate the parties and edges consecutively, i.e., $e_j=(e_j^j,e_j^{j+1})\in E$. First, party $N$ applies the generalized measurement
\begin{equation}
\{1/4 h_N[(A^{-1}\sigma_i)_{e_N^{N}}\otimes (\sigma_k)_{e_{N-1}^{N}}]\},
\end{equation}
with $\tr (A^{-\dagger}H_NA^{-1})=4$ for normalization. Depending on the measurement outcome, party $1$ and party $N-1$ apply $\sigma_i$ and $\sigma_k$ respectively on the qubit that is connected to party $N$. Then, beginning with party $N-1$ and proceeding in decreasing order, party $j\in\{3,\ldots, N\}$ applies the generalized measurement
\begin{equation}
\{1/2 h_j(\one_{e_j^{j}}\otimes (\sigma_i)_{e_{j-1}^{j}})\}
\end{equation}
and for outcome $i$ party $j-1$ applies $\sigma_i$ on $e_{j-1}^{j-1}$ before it measures in the subsequent round. Finally, party $2$ applies the measurement
\begin{equation}
\{1/2 h_2(\one_{e_2^{2}}\otimes (B^{-1}\sigma_i)_{e_{1}^{2}})\}
\end{equation}
and for outcome $i$ party $1$ implements $\sigma_i$. With this, the initial state in the final round is given by  $A\otimes B \otimes_{i=2}^N h_i \ket{\Psi}$. In the last round party $1$ applies the generalized measurement $\{h_1 S_i^{(1)}\}$ which is a valid POVM due to Eq. (\ref{finiteLOCCk}). For outcome $i$, all other parties apply the unitary $V_i^{(j)}$ such that $V_i^{(j)} h_j=h_jS_i^{(j)}$. These unitaries exist due to Eq. (\ref{commute}). This protocol allows to reach the state $h\ket{\Psi}$ deterministically.

Unfortunately, however, this provides no complete characterization of finite round LOCC transformations in cycle networks, as the following example illustrates.

\section{Example of a transformation that requires  one party to measure more than once}

Analogous to the case of unitary symmetries, it may be that probabilistic steps are required \cite{finiteroundLU,finiteroundLU2}. Moreover, even if one restricts to protocols for which each LOCC round corresponds to a deterministic transformation among pure states (and therefore Eqs.~\eqref{commute} and~\eqref{finiteLOCCk} have to hold true for each round), it is not straightforward to find these protocols as parties may be required to measure more than once. In the simplest case of transformation between bipartite states with pure initial state, it is known that one round of communication is always sufficient~\cite{Lo2001,Nielsen}. For more parties, more rounds may be required. The minimum number of rounds is also called the \emph{round complexity}, and it it is known that exist tasks, including state discrimination and entanglement distillation, that show a separation between LOCC protocols with different round complexity~\cite{WinterLOCC}.

In case of transformations among generic four-qubit states~\cite{Vicente2013}, it can be easily seen that certain transformations require a round complexity of two. Moreover, the following transformation provides, up to our knowledge, the first example of a pure state LOCC transformation among fully entangled states that requires a round complexity of three. For this transformation, which requires one party to measure more than once, we consider the triangle network and the SLOCC operators
\begin{equation}
H_j= \one + \sum_{i=x,y} a_i^{(j)}\sigma_i\otimes \sigma_i
\end{equation}
for $j\in\{1, 2\}$ and
\begin{equation}
H_3= U_1^\dagger\otimes U_1^\dagger(\one + c \sigma_z\otimes \sigma_z)U_1\otimes U_1
\end{equation}
with $a_i^{(j)}\neq 0$, $a_x^{(j)}\neq \pm a_y^{(j)}$, $c\neq 0$ and $U_1=e^{i \alpha_x \sigma_x}$ with $\alpha_x\neq 0, m \pi/4$ and $m$ being an integer. In particular, we will show that for all symmetries for which Eq.~\eqref{commute} is fulfilled for two parties, Eq.~\eqref{finiteLOCCk} with $g_k=\one$ (i.e., party $k$ has not performed a measurement yet) does not hold true for the remaining party $k$. Thus, this state cannot be reached from the initial network state via a protocol where each party measures solely once. If party 3 is allowed to measure twice, the state can be reached, and we will provide an explicit protocol achieving this transformation.

We will make use of the necessary conditions in Eqs. (\ref{commute}) and (\ref{finiteLOCCk}) and  the systematic approach explained in Section~\ref{sec:systematic} that allows to identify the local operators which fulfill Eq. (\ref{commute}). Note that for party $1$ and $2$ $H_j$ is already Bell diagonal and corresponds to the non-degenerate case, which implies that for party $1$ and $2$ the symmetries that can potentially be used locally in the last round are $\sigma_i\otimes \sigma_i$, and for party $3$ they are
\begin{equation}
U_1^\dagger\otimes U_1^\dagger \{[Z(\phi_1)\otimes Z(\phi_2)](\sigma_x\otimes\sigma_x)^k\}U_1\otimes U_1
\end{equation}
with $k\in \{0,1\}$ and $Z(\phi)= diag(e^{i \phi},e^{-i \phi})$ as the local standard form has two double degenerate eigenvalues (see Appendix~\ref{Appendix_B}). 

The only operators $\otimes_{j\neq k} X_j\otimes Y_j$ on two parties which have a non-zero intersection with a symmetry restricted to two parties are (a) $ \sigma_i^{\otimes 4}$ on party $1$ and $2$ (with corresponding symmetries  $\sigma_i^{\otimes 6}$) or (b) $U_1^\dagger Z(\phi)  U_1\otimes\one^{\otimes 3} $ and  $\sigma_x^{\otimes 4}$ on party $3$ and one of the other parties (with corresponding symmetries $U_1^\dagger Z(\phi)  U_1\otimes  \one^{\otimes 4} \otimes U_1^\dagger Z(\phi) U_1$ and $\sigma_x^{\otimes 6}$). 
 
Let us first consider case (a) in which the party measuring in the last round is party $3$. It is straightforward to see that the final state is not reachable in a protocol where this party has not performed any measurement before, and therefore $g_3=A\otimes B$. This is due to the fact that
\begin{equation}
\sum_i p_i \sigma_i\otimes\sigma_i H_3\sigma_i\otimes \sigma_i \not \propto G_3=A^{\dagger}A\otimes B^{\dagger}B
\end{equation}
for any probability distribution $p_i$ (see Observation~\ref{obs:observation3}). Analogously, we have that for case (b)
\begin{equation}
\sum_i p_i U_1^\dagger Z(-\phi_i)  U_1\otimes\one H_j U_1^\dagger Z(\phi_i)  U_1\otimes\one\not \propto A^{\dagger}A\otimes B^{\dagger}B
\end{equation}
for $j\in\{1, 2\}$ and therefore also in this case, it is not possible to reach the final state if each party measures only once. Note here that the symmetry $\sigma_x^{\otimes 6}$ does not need to be taken into account as  $\sigma_x^{\otimes 6}$ commutes with $H$ and hence is equivalent to the identity in our analysis.

However, it can be reached via the following protocol. First party $3$ measures
\begin{equation}
\{\frac{1}{\sqrt{2}}\tilde{h}_3, \frac{1}{\sqrt{2}}\tilde{h}_3(\one\otimes \sigma_y)\},
\end{equation}
with $\tilde{h}_3=\sqrt{\one + \tilde{c} \sigma_z\otimes\sigma_z}$ and $\tilde{c}=1/4 \tr (H_3\sigma_z\otimes \sigma_z)$ which corresponds to a valid POVM. In case of the second outcome party 2 applies $\sigma_y$ on the qubit which is entangled with the second qubit of party 3. With this they reach deterministically the state $(\one\otimes\one\otimes\tilde{h}_3)(\bigotimes_{e\in E} \ket{\Phi^+}_e)$. Then party 1 and 2 measure
\begin{equation}
\{\frac{1}{\sqrt{2}}h_j, \frac{1}{\sqrt{2}}h_j(\one\otimes\sigma_z)\},
\end{equation}
where for the second outcome, $\sigma_z$ is acting on the qubit that is connected to the third party and in case this outcome is obtained, party 3 applies $\sigma_z$ on the qubit which originated from the same source. By doing so, they obtain deterministically the state $(h_1\otimes h_2 \otimes \tilde{h}_3)(\bigotimes_{e\in E} \ket{\Phi^+}_e)$. Finally, party 3 implements the measurement
\begin{equation}
\{\frac{1}{\sqrt{2}} h_3(\sigma_i\otimes\sigma_i)\tilde{h}_3^{-1}\},
\end{equation}
with $i=0, z$. This is a valid POVM and party $j=1,2$ apply for outcome $i$ the unitary $V_{ij}$ with the property that $V_{ij}h_j=h_j(\sigma_i\otimes\sigma_i)$ (which have to exist as $[H_j,(\sigma_i\otimes\sigma_i)]=0$). This transformation deterministically reaches the desired state.

\section{Conclusion and Outlook}
In this work, we have studied finite-round LOCC transformation in networks. We first provided the symmetries for arbitrary network structures, if the sources distribute two-qubit states. Which finite-round LOCC transformations can be implemented is determined by the symmetries \cite{Gour,Gour2,measuroutcdonotexist,finiteroundgeneral}, and we use this to show that (in contrast to tree networks \cite{Hayata}) for networks containing a cycle there exist states which are not reachable. For cycle networks, we provide a systematic way to determine classes of states that cannot be reached via finite-round LOCC transformations. Moreover, we characterize all states which can be reached with protocols where each party measures once, and each step of the protocol is deterministic.
Finally, we provide an example of a state which is not reachable with simple protocols where each party measures solely once, showing that in general more involved protocols have to be considered.

In this work we have only considered qubit sources. One may generalize our approach to derive the symmetries to higher-dimensional systems, and use this to characterize finite-round LOCC transformations in networks distributing qudit states. Furthermore, the scenario we consider, namely a network distributing single copies of bipartite states accompanied with finite-round LOCC transformations, is highly relevant for quantum communication tasks and of practical relevance. Hence, it would be desirable to further study the resource theory of network entanglement and identify operationally meaningful entanglement measures.\\

\acknowledgements

We thank Hayata Yamasaki, Martin Hebenstreit, Barbara Kraus and Otfried G\"uhne for discussions and comments. This work has been supported by the the Austrian Science Fund (FWF): J 4258-N27, Y879-N27 and P 32273-N27, the Austrian Academy of Sciences, the ERC (Consolidator Grant 683107/TempoQ), and the Deutsche Forschungsgemeinschaft (DFG, German Research Foundation - 447948357 and 440958198).

\appendix
\onecolumngrid
\renewcommand\theequation{A\arabic{equation}}
\setcounter{equation}{0}

\section{Proof of Lemma \ref{Lemma_factorize}} \label{Appendix_A}
In the following, we provide a proof for the fact that the local symmetry operators $S_2$ and $S_3$ in Eq. (\ref{eq_symmetry}) in the main text have to factorize.
\begin{lemma}\label{Lemma_factorize}
The relation $X_{AB}Z_{BC}=Y_{AC}\otimes \one_B$ with system $B$ being two-dimensional, and $X, Z$ invertible can only hold true if $X_{AB}=\bar{X}_A\otimes X_B$ and $Z_{BC}=\bar{Z}_B\otimes Z_C$ for some choice of invertible operators $\bar{X}_A, X_B, \bar{Z}_B$ and $Z_C$, i.e., the operators factorize.
\end{lemma}

\begin{proof}
We will show that $X_{AB}Z_{BC}=Y_{AC}\otimes \one_B$, with system $B$ being two-dimensional, and $X, Z$ invertible implies that $X_{AB}=\tilde{X}_A\otimes X_B$ and $Z_{BC}=\tilde{Z}_B\otimes Z_C$. As we have explained in the main text (see also Fig.~\ref{fig:symm} in the main text), this implies that in any network the local symmetries have to factorize with respect to the different particles the parties receive from the sources.

So let us consider the equation 
\begin{equation}\label{eq_fact}
X_{AB}Z_{BC}=Y_{AC}\otimes \one_B.
\end{equation}
Note that this equation holds true if and only if 
\begin{equation}\label{eq_fact2}
\tilde{X}_{AB}\tilde{Z}_{BC}=Y_{AC}\otimes \one_B,
\end{equation}
with $\tilde{Q}=M_B^{-1}Q M_B$ for $Q=X, Z$. Note further that $Q$ factorizes with respect to the splitting B versus rest if and only if $\tilde{Q}$ does.
We will decompose the operators $\tilde{X}$ and $\tilde{Z}$ in an operator basis $\{\lambda_i\otimes \sigma_j\}$, where $\sigma_j$ is acting on system B and correspond to the Pauli operators, i.e.,
\begin{equation}
\tilde{X}_{AB}=\sum _{i,j} x_{ij} \lambda_i\otimes \sigma_j
\end{equation}
and
\begin{equation}
\tilde{Z}_{BC}=\sum _{k,l} z_{kl} \sigma_k\otimes \lambda_l=\sum_l W_l\otimes \lambda_l.
\end{equation}
Moreover, we will choose M such that $W_0=\sum _{k} z_{k0} \sigma_k$ is in Jordan decomposition. Here we chose without loss of  generality to bring the matrix $W_0$ acting on B for the component $\lambda_0$ into Jordan form, however, in case this component does vanish we relabel the basis for C such that some non-zero component is referred to as $\lambda_0$.
The matrix $W_0$ can then either be diagonal or have a Jordan block of size 2. We will first discuss the diagonal case.
Then in order for Eq. (\ref{eq_fact2}) to hold true we have that for all $i$ and $j\in{1,2,3}$
\begin{align}\nonumber
&\tr [(\tilde{X}_{AB}\tilde{Z}_{BC})(\lambda_i\otimes \sigma_j\otimes \lambda_0)]=\tr [(\tilde{X}_{AB}\{\one_A\otimes W_0\})(\lambda_i\otimes \sigma_j)]=0.
\end{align}
This is equivalent to
\begin{align}\nonumber
&x_{i1}z_{00}+i x_{i2} z_{30}=0,\\\nonumber
&x_{i2}z_{00}-i x_{i1} z_{30}=0,\\
&x_{i3}z_{00}+x_{i0}z_{30}=0.
\label{eq_cond}
\end{align}
These equations can also be written as
\begin{align}\label{vectors}
\vec{v}^{j}_i\cdot \vec{\omega}=0,
\end{align}
with $\vec{\omega}= (z_{00},z_{03})^T$, $\vec{v}^{1}_i=(x_{i1},i x_{i2})^T$,$\vec{v}^{2}_i=(x_{i2},-i x_{i1})^T$ and $\vec{v}^{3}_i=(x_{i3},ix_{i0})^T$. From Eq. (\ref{vectors}) for $j=1$ and $j=2$ we have that  $\vec{v}^{2}_i= c_i \vec{v}^{2}_i$ with $c_i$ being some proportionality factor, which can be straightforwardly be shown to be $\pm i $. With this we have that $x_{i2}=\pm i x_{i1}$. In case $x_{i1}\neq 0$ it therefore follows that $z_{00}=\pm z_{03}$. This implies that $x_{i2}=\pm i x_{i1}$ for all $i$, i.e. it is the same sign for all $i$, and also $x_{i3}=\mp x_{i0}$ for all $i$. 
It should be noted that these constraints lead to a matrix $W_i$ which has a kernel. The corresponding eigenvector is $(1,0)^T$ for $x_{i2}= i x_{i1}$ or $(0,1)^T$ for  $x_{i2}=-i x_{i1}$. As any $W_i$ has the same kernel this solution is only possible if $\tilde{Z}$ is not invertible (which contradicts our assumption). Hence, it has to hold that $x_{i2}=x_{i1}=0$. From Eq. (\ref{vectors}) for $j=3$ it follows that all vectors $\vec{v}^{3}_i$ have to be parallel (or the zero vector), which corresponds to the case that $\tilde{X}$ factorizes. It can be straightforwardly seen that in this case also $\tilde{Z}$ has to factorize.

It remains to consider the case that $W_0$ is a Jordan block of size 2, i.e. $z_{30}=0$ and $z_{20}=i z_{10}=i/2$. Then Eq. (\ref{eq_fact}) is fulfilled  only if for all $i$
\begin{align}\nonumber
&x_{i1}z_{00}+i
(x_{i0}+x_{i3}) z_{10}=0,\\\nonumber
&x_{i2}z_{00}+i (x_{i0}+ x_{i3})z_{10}=0,\\
&x_{i3}z_{00}+ i(i x_{i1}-x_{i2})z_{10}=0.
\label{eq_cond2}
\end{align}
Let us first consider the case that $z_{00}=0$. Then for all $i$ it has to hold that $x_{i2}=i x_{i1}$ and  $x_{i0}=- x_{i3}$. As before this would imply in contradiction to our assumption that $\tilde{Z}$ is not invertible.
Considering the case $z_{00}\neq 0$ it can be easily seen that  $x_{i2}=i x_{i1}$ and $x_{i3}=0$. We can then write the first condition in Eq.~(\ref{eq_cond2}) as
\begin{align}\nonumber
0=x_{i1}z_{00}+i x_{i0}z_{10}\equiv \vec{r}_i\cdot\vec{s},
\end{align}
with $ \vec{r}_i=(x_{i1},x_{i0})^T$ and $\vec{s}=(z_{00},i z_{10})^T$. Hence, for all $i$ the vectors $ \vec{r}_i$ have to be parallel (or the zero vector). This implies that $\tilde{X}$ and hence also $\tilde{Z}$ has to factorize, which completes the proof that in order for Eq. (\ref{eq_fact}) to hold $X$ and $Z$ have to factorize.
 \end{proof}

\section{Symmetries for the local standard form}
\label{Appendix_B}
The local standard form is chosen such that $H_j^{mb} $ is diagonal with real positive entries. Which symmetries are possible will depend on the degeneracy. Note that one can always permute the diagonal entries of  $H_j^{mb}$ to bring it in a form mentioned below, e.g., one can permute $\proj{\Psi^+}$ and $\proj{\Psi^-}$ (leaving at the same time $\proj{\Phi^\pm}$ invariant) by conjugation with $\sqrt{i\sigma_z}\otimes \sqrt{-i\sigma_z}$. Via the procedure outlined in the main text one can determine the symmetries of $H_j^B$. We obtain the following results:

\begin{enumerate}[label=(\roman*)]
\item No degeneracy, i.e., $H_j^{mb} =diag(a,b,c,d)$: The local operators which leave $H_j^B$ invariant are $\sigma_i\otimes\sigma_i$ (see also~\cite{Spee2016}).
\item A double-degenerate and two different eigenvalues, i.e., $H_j^{mb} =diag(a,a,b,c)$: The operators  $[Z(\phi)\otimes Z(\phi)](\sigma_x\otimes\sigma_x)^k$ with $Z(\phi)= diag(e^{i \phi},e^{-i \phi})$ and $k\in \{0,1\}$ are the only local ones which leave $H_j^B$ invariant.
\item Two double-degenerate eigenvalues, i.e., $H_j^{mb} =diag(a,a,b,b)$: The symmetries of $H_j^B$ are $[Z(\phi_1)\otimes Z(\phi_2)](\sigma_x\otimes\sigma_x)^k$ with $k\in \{0,1\}$.
\item Triple degenerate eigenvalues, i.e., $H_j^{mb} =diag(a,a,a,b)$: The local  symmetries of $H_j^B$  are given by $Z(\phi_1)X(\alpha)Z(\phi_2)\otimes Z(\phi_1)X(-\alpha)Z(\phi_2)$ with $X(\alpha)=e^{i\alpha\sigma_x}$.
\item Four degenerate eigenvalues, i.e., $H_j^{mb} =\one$: All local unitaries leave $H_j^B$ invariant.
\end{enumerate}


\begin{thebibliography}{99}
\bibitem{Kimble2008} H. J. Kimble, \href{\doibase 10.1038/nature07127}{Nature \textbf{453}, 1023 (2008)}.

\bibitem{Wehner2018} S. Wehner, D. Elkouss and R. Hanson, \href{\doibase 10.1126/science.aam9288}{Science \textbf{362}, 9288 (2018)}.

\bibitem{Cirac1997} J. I. Cirac, P. Zoller, H. J. Kimble, and H. Mabuchi, \href{\doibase 10.1103/PhysRevLett.78.3221}{Phys. Rev. Lett. \textbf{78}, 3221 (1997)}.

\bibitem{Duan2010} L.-M. Duan and C. Monroe, \href{\doibase 10.1103/RevModPhys.82.1209}{Rev. Mod. Phys. \textbf{82}, 1209 (2010)}.

\bibitem{Reiserer2015} A. Reiserer and G. Rempe, \href{\doibase 10.1103/RevModPhys.87.1379}{Rev. Mod. Phys. \textbf{87},1379 (2015)}.

\bibitem{Duan2001} L.-M. Duan, M. D. Lukin, J. I. Cirac, P. Zoller, \href{\doibase 10.1038/35106500}{Nature \textbf{414}, 413 (2001)}.

\bibitem{Cirac1999} J. I. Cirac, A. K. Ekert, S. F. Huelga, and C. Macchiavello, \href{\doibase 10.1103/PhysRevA.59.4249}{Phys. Rev. A \textbf{59}, 4249 (1999)}. 


\bibitem{Spiller2006} T. P. Spiller, K. Nemoto, S. L. Braunstein, W. J. Munro, P. van Loock, and G. J. Milburn, \href{\doibase 10.1088/1367-2630/8/2/030}{New J. Phys. \textbf{8}, 30 (2006)}.

\bibitem{Elkouss2020} K. Azuma, S. Bäuml, T. Coopmans, D. Elkouss, B. Li,	\href{https://doi.org/10.1116/5.0024062}{AVS Quantum Sci. \textbf{3}, 014101 (2021)}.

\bibitem{Gisin2020} N. Gisin, J.-D. Bancal, Y. Cai, P. Remy, A. Tavakoli, E. Zambrini Cruzeiro, S. Popescu, N. Brunner, \href{\doibase 10.1038/s41467-020-16137-4}{Nat. Commun. \textbf{11}, 2378 (2020)}.

\bibitem{Kraft2020triangle} T. Kraft, S. Designolle, C. Ritz, N.Brunner, O. G\"uhne, and M. Huber, \href{\doibase 10.1103/PhysRevA.103.L060401}{Phys. Rev. A. \textbf{103}, L060401 (2021)}.

\bibitem{Navascues2020} M. Navascués, E. Wolfe, D. Rosset, and A. Pozas-Kerstjens, \href{\doibase 10.1103/PhysRevLett.125.240505}{Phys. Rev. Lett. \textbf{125}, 240505 (2020)}.

\bibitem{Luo2020} M.-X. Luo, \href{\doibase 10.1002/qute.202000123}{Adv. Quantum Technol., 2000123 (2021)}.

\bibitem{Aberg2020} J. Åberg, R. Nery, C. Duarte, R. Chaves, \href{\doibase 10.1103/PhysRevLett.125.110505}{Phys. Rev. Lett. \textbf{125}, 110505 (2020)}.

\bibitem{Kraft2020} T. Kraft, C. Spee, X.-D. Yu, and O. G\"uhne, \href{\doibase 10.1103/PhysRevA.103.052405}{Phys. Rev. A \textbf{103}, 052405 (2021)}.

\bibitem{Hansenne2022} K. Hansenne, Z.-P. Xu, T. Kraft, and O. G\"uhne, \href{https://doi.org/10.1038/s41467-022-28006-3}{Nat. Commun. \textbf{13}, 496 (2022)}.

\bibitem{teleport} C. H. Bennett, G. Brassard, C. Cr\'epeau, R. Jozsa, A. Peres, and W. K. Wootters, \href{\doibase 10.1103/PhysRevLett.70.1895}{Phys. Rev. Lett. \textbf{70}, 1895 (1993)}.

\bibitem{Acin2007} A. Acín, J. Cirac, M. Lewenstein, \href{\doibase 10.1038/nphys549}{Nature Physics \textbf{3}, 256 (2007)}.

\bibitem{Nielsen} M. A. Nielsen, \href{\doibase 10.1103/PhysRevLett.83.436}{Phys. Rev. Lett. \textbf{83}, 436 (1999)}.

\bibitem{Duer} W. D\"ur, G. Vidal, and J.I. Cirac, \href{\doibase 10.1103/PhysRevA.62.062314}{Phys. Rev. A \textbf{62},062314 (2000)}.

\bibitem{Verstraete} F. Verstraete, J. Dehaene, B. De Moor, and H. Verschelde, \href{\doibase 10.1103/PhysRevA.65.052112} {Phys. Rev. A \textbf{65}, 052112 (2002)}.

\bibitem{Donald2002}M. J. Donald, M. Horodecki, and O. Rudolph, \href{\doibase 10.1063/1.1495917}{J. Math. Phys. \textbf{43}, 4252 (2002).}

\bibitem{Chitambar2011} E. Chitambar, \href{\doibase 10.1103/PhysRevLett.107.190502}{Phys. Rev. Lett. \textbf{107}, 190502 (2011)}.

\bibitem{Chitambar2012}E. Chitambar, W. Cui, and H.-K-. Lo, \href{\doibase 10.1103/PhysRevLett.108.240504}{Phys. Rev. Lett. \textbf{108}, 240504 (2012)}.

\bibitem{WinterLOCC}  E. Chitambar, D. Leung, L. Mancinska, M. Ozols, A. Winter, \href{\doibase 10.1007/s00220-014-1953-9}{Commun. Math. Phys. \textbf{328}, 303 (2014)}.

\bibitem{Cohen2017} S. M. Cohen,  \href{\doibase 10.1103/PhysRevLett.118.020501}{Phys. Rev. Lett. \textbf{118}, 020501 (2017)}.

\bibitem{Turgut1} S. Turgut, Y. G\"ul, and N. K. Pak, \href{\doibase  10.1103/PhysRevA.81.012317}{Phys. Rev. A \textbf{81}, 012317 (2010).}

\bibitem{Turgut2}S. Kintas and S. Turgut,  \href{\doibase  10.1063/1.3481573}{J. Math. Phys.
\textbf{51}, 092202 (2010).}

\bibitem{finiteroundLU} C. Spee, J.I. de Vicente, D. Sauerwein, B. Kraus, \href{\doibase  10.1103/PhysRevLett.118.040503}{Phys. Rev. Lett. \textbf{118}, 040503 (2017)}.

\bibitem{finiteroundLU2} J.I. de Vicente, C. Spee, D. Sauerwein, B. Kraus, \href{\doibase  10.1103/PhysRevA.95.012323}{Phys. Rev. A \textbf{95}, 012323 (2017)}.

\bibitem{Vicente2013} J. I. de Vicente, C. Spee, and B. Kraus, \href{\doibase 10.1103/PhysRevLett.111.110502}{Phys. Rev. Lett. 111, 110502 (2013)}.

\bibitem{Schwaiger2015} K. Schwaiger, D. Sauerwein, M. Cuquet, J. I. de Vicente, B. Kraus, \href{\doibase 10.1103/PhysRevLett.115.150502}{Phys. Rev. Lett. \textbf{115}, 150502 (2015)}.

\bibitem{Spee2016} C. Spee, J. I. de Vicente, B. Kraus, \href{\doibase 10.1063/1.4946895}{J. Math. Phys. \textbf{57}, 052201 (2016)}.

\bibitem{Hebenstreit2016}M. Hebenstreit, C. Spee, and B. Kraus, \href{\doibase 10.1103/PhysRevA.93.012339}{Phys. Rev. A \textbf{93}, 012339 (2016)}.

\bibitem{Hayata} H. Yamasaki, A. Soeda, and M. Murao, \href{\doibase 10.1103/PhysRevA.96.032330}{Phys. Rev. A \textbf{96}, 032330 (2017)}.

\bibitem{measuroutcdonotexist} M. Hebenstreit, M. Englbrecht, C. Spee, J. I. de Vicente, and  B. Kraus, \href{\doibase 10.1088/1367-2630/abe60c}{New J. Phys. \textbf{23}, 033046 (2021)}.

\bibitem{Gour} G. Gour and N. R. Wallach, \href{\doibase 10.1088/1367-2630/13/7/073013}{New J. Phys. \textbf{13}, 073013 (2011)}.

\bibitem{Gour2} G Gour and N. R. Wallach, \href{\doibase 10.1088/1367-2630/ab4c88}{New J. Phys. \textbf{21}, 109502 (2019)}.

\bibitem{MPS} D. Sauerwein, A. Molnar, J. I. Cirac, and B. Kraus, \href{\doibase 10.1103/PhysRevLett.123.170504}{Phys. Rev. Lett. \textbf{123}, 170504 (2019)}.

\bibitem{finiteroundgeneral} M. Hebenstreit, C. Spee, N. K. H. Li, B. Kraus, J. I. de Vicente, \href{https://doi.org/10.1103/PhysRevA.105.032458}{Phys. Rev. A 105, 032458 (2022)}.

\bibitem{Belldiagonal} F. Verstraete, J. Dehaene, and B. De Moor, \href{\doibase 10.1103/PhysRevA.64.010101}{Phys. Rev. A \textbf{64}, 010101(R) (2001)}.

\bibitem{Lo2001} H.-K. Lo and S. Popescu, \href{\doibase doi:10.1103/PhysRevA.63.022301}{Phys. Rev. A, \textbf{63}, 022301 (2001)}.


\end{thebibliography}
\end{document}